\RequirePackage{amsmath}
\RequirePackage{amsthm}

\documentclass[runningheads]{llncs}

\usepackage{amssymb}    %
\usepackage{amsfonts}
\usepackage{graphicx}   %
\usepackage{wrapfig}    %
\usepackage[backend=bibtex]{biblatex} %
\usepackage{xcolor}     %
\usepackage{hyperref}   %
\usepackage[capitalize]{cleveref} %
\usepackage[colorinlistoftodos]{todonotes} %
\usepackage{listings}   %
\usepackage{framed}     %
\usepackage[nounderscore]{syntax} %
\usepackage{semantic}   %
\usepackage[inline]{enumitem} %
\usepackage{verbatim}   %
\usepackage{etoolbox}   %
\usepackage{syntax}     %
\usepackage{stmaryrd}   %
\usepackage{cancel}       %
\usepackage{adjustbox}  %
\usepackage{cutwin}     %
\usepackage{fontawesome} %

\newcommand{\orcid}[1]{}

\newcommand{\framegather}[0]{\vspace*{-0.3cm}}

\usepackage{tikz}
\usepackage{tikz-cd}    %

\hypersetup{pdfpagemode=UseOutlines, colorlinks=true, linkcolor=blue, citecolor=blue}

\BeforeBeginEnvironment{wrapfigure}{\setlength{\intextsep}{0pt}}

\newcommand{\nat}{\Longrightarrow} %
\newcommand{\pr}{^\prime}
\newcommand{\prpr}{^{\prime\prime}}
\newcommand{\init}{\bot} %
\newcommand{\term}{\top} %

\newcommand{\cat}[1]{{\mathbb{#1}}}
\newcommand{\Set}{\mathbf{Set}}

\newcommand{\Id}{\mathrm{Id}}
\newcommand{\id}{\mathrm{id}}

\newcommand{\hole}{\mathrm{hole}}
\newcommand{\con}{\mathrm{con}}

\newcommand{\Proj}{\pi}
\newcommand{\Inj}{i}
\newcommand{\strength}{\mathrm{st}}

\newcommand{\plug}[1]{\left\llbracket #1 \right\rrbracket}

\newcommand{\bisim}{\sim_{\mathrm{bis}}}
\newcommand{\bisimspan}{(\bisim)}

\newif\ifedit

\ifedit
\newcommand{\ST}[1]{\textcolor{purple}{ST: #1}}
\newcommand{\STin}[1]{\todo[color=purple!30,inline]{Stelios: #1}}
\newcommand{\AN}[1]{\textcolor{brown}{AN: #1}}
\newcommand{\ANin}[1]{\todo[color=yellow!30,inline]{Andreas: #1}}
\newcommand{\DD}[1]{\textcolor{brown}{DD: #1}}
\newcommand{\DDin}[1]{\todo[color=yellow!30,inline]{Dominique: #1}}
\else
\newcommand{\ST}[1]{}
\newcommand{\STin}[1]{}
\newcommand{\AN}[1]{}
\newcommand{\ANin}[1]{}
\newcommand{\DD}[1]{}
\newcommand{\DDin}[1]{}
\fi

\newcommand{\arrow}[2]{%
  $\begin{tikzcd}[ampersand replacement=\&] #1 \arrow[r, shift left=0.3ex] \& #2 \end{tikzcd}$}

\newcommand{\goes}[2]{\ensuremath{#1 \rightarrow #2}}

\newcommand{\rets}[2]{\ensuremath{#1 \Downarrow #2}}

\newcommand{\goesv}[3]{\ensuremath{#1 \rightarrow_{#3} #2}}

\newcommand{\retsv}[3]{\ensuremath{#1 \Downarrow_{#3} #2}}

\newcommand{\retsvr}[3]{\ensuremath{#1~\textcolor{red}{\Downarrow_{#3}}~#2}}

\newcommand{\while}{\emph{While}}

\newcommand{\flagsym}{{\text{\faBellO}}}

\newcommand{\secflagsym}{{\text{\faBellSlashO}}}
\newcommand{\whilep}{$\mathit{While}_{\flagsym}$}
\newcommand{\whiles}{$\mathit{While}_{\secflagsym}$}
\newcommand{\intsym}{{\mathbb{Z}}}
\newcommand{\whilez}{$\mathit{While}_{\intsym}$}
\newcommand{\whileb}{$\mathit{While}_{B}$}
\newcommand{\low}{$\mathit{Low}$}
\newcommand{\whilestack}{$\mathit{Stack}$}

\newcommand{\set}{\mathbf{Set}}

\bibliography{mainBiblio}

\begin{document}
\title{A categorical approach to secure compilation}

\author{Stelios Tsampas\inst{1}\orcid{0000-0001-8981-2328} \and
Andreas Nuyts\inst{1}\orcid{0000-0002-1571-5063} \and
Dominique Devriese\inst{2}\orcid{0000-0002-3862-6856} \and
Frank Piessens\inst{1}\orcid{0000-0001-5438-153X}}

\authorrunning{Tsampas et al.}

\institute{KU Leuven, Leuven, Belgium \email{name.surname@cs.kuleuven.be}\and
Vrije Universiteit Brussel, Brussels, Belgium
\email{dominique.devriese@vub.be}}
\maketitle              %
\begin{abstract}
We introduce a novel approach to secure compilation based on maps of
distributive laws. We demonstrate through four examples that the 
coherence criterion for maps of distributive laws can potentially be a viable
alternative for compiler security instead of full abstraction, which is the
preservation and reflection of contextual equivalence. To that end, we also make
use of the well-behavedness properties of distributive laws to construct a
categorical argument for the contextual connotations of bisimilarity.

\keywords{Secure compilation \and Distributive laws \and Structural Operational Semantics}
\end{abstract}

\section{Introduction}
\label{sec:intro}

As a field, secure compilation is the study of compilers that formally preserve
abstractions across languages. Its roots can be tracked back to the seminal
observation made by Abadi~\cite{DBLP:conf/ecoopw/Abadi99}, namely that compilers
which do not protect high-level abstractions against low-level contexts might
introduce security vulnerabilities. But it was the advent of secure
architectures like the Intel 
SGX~\cite{DBLP:journals/iacr/CostanD16} and an ever-increasing need for computer
security that motivated researchers to eventually work on formally proving
compiler security.

The most prominent~\cite{DBLP:conf/popl/DevriesePP16,
  DBLP:conf/popl/FournetSCDSL13, van_strydonck_linear_2019,
  DBLP:conf/csfw/PatrignaniDP16,
  DBLP:journals/toplas/PatrignaniAS0CP15, DBLP:conf/icfp/NewBA16, skorstengaard_stktokens:_2019} formal 
criterion for compiler security is \emph{full
  abstraction}: A compiler is fully abstract if it preserves and reflects
Morris-style contextual equivalence~\cite{morris}, i.e. indistinguishability
under all program contexts, which are usually defined as programs with a hole.
The intuition is that contexts represent the ways an attacker can interact with
programs and so full abstraction ensures that such interactions are
consistent between languages.

Full abstraction is arguably a strong and useful property but it is also
notoriously hard to prove for realistic compilers, mainly due to the
inherent challenge of having to reason directly about program
contexts~\cite{DBLP:journals/toplas/PatrignaniAS0CP15,
  DBLP:conf/popl/FournetSCDSL13, DBLP:conf/icfp/AhmedB11,
  DBLP:conf/csfw/JagadeesanPRR11}. There is thus a need for better formal
methods, a view shared in the scientific
community~\cite{ahmed_et_al:DR:2018:9891, Patrignani2019FormalAT}. While recent
work has proposed generalizing from full abstraction towards the so-called \emph{robust}
properties~\cite{DBLP:conf/esop/PatrignaniG19, abate2018journey}, the main
challenge of quantifying over program contexts remains, which manifests when
directly translating target contexts to the source (\emph{back-translation}).
Other techniques, such 
as trace semantics~\cite{DBLP:conf/csfw/PatrignaniDP16} or logical
relations~\cite{DBLP:journals/corr/abs-1103-0510}, require complex correctness
and completeness proofs w.r.t. contextual equivalence in order to be applicable.

In this paper we introduce a novel, categorical approach to secure compilation.
The approach has two main components: the elegant representation of Structural
Operational Semantics (SOS)~\cite{DBLP:journals/jlp/Plotkin04a} using
category-theoretic \emph{distributive laws}~\cite{DBLP:conf/lics/TuriP97}~\footnote{The
  authors use the term ``Mathematical Operational Semantics''. The term
  ``Bialgebraic Semantics'' is also used in the literature.} and also \emph{maps of
distributive laws}~\cite{DBLP:journals/entcs/PowerW99,
DBLP:journals/entcs/Watanabe02, DBLP:conf/calco/KlinN15} as secure compilers
that preserve\ST{ and possibly reflect} bisimilarity. Our method aims to be
unifying, in that there is a general, shared formalism for operational
semantics, and simplifying, in that the formal criterion for compiler security,
the \emph{coherence criterion} for maps of distributive laws, is straightforward
and relatively easy to prove.

The starting point of our contributions is an abstract proof on how
coalgebraic bisimilarity under distributive laws holds \emph{contextual} meaning
in a manner similar to contextual equivalence (\Cref{sec:proof}). We argue that
this justifies the use of the coherence criterion for testing compiler security
as long as bisimilarity adequately captures the underlying threat model.
We then demonstrate the effectiveness of our approach by appeal to four
examples of compiler (in)security. The examples
model classic, non-trivial problems in secure compilation: %
\begin{itemize}
  \item An example of an extra processor register in the target language that
    conveys additional information about computations (\Cref{sec:ex1}).
  \item A datatype mismatch between the type of variable (\Cref{sec:state}).
  \item The introduction of illicit control flow in the target language (\Cref{sec:control})
  \item A case of incorrect local state encapsulation (\Cref{sec:local}).
\end{itemize}
For each of these examples we present an insecure compiler that fails the
coherence criterion, then introduce \emph{security primitives} in the target
language and construct a secure compiler that respects it. We also
examine how bisimilarity can be both a blessing and a curse as its strictness
and rigidity sometimes lead to relatively contrived solutions. Finally,
in~\Cref{sec:proscons}, we discuss related work and point out potential
avenues for further development of the underlying theory.

\subsubsection{On the structure and style of the paper}

This work is presented mainly in the style of programming language semantics but
its ideas are deeply rooted in category theory. We follow an ``on-demand''
approach when it comes to important categorical concepts: we begin the first
example by introducing the base language used throughout the paper, \while{},
and gradually present distributive laws when required. From the second
example in~\Cref{sec:state} and on, we relax the categorical notation and mostly
remain within the style of PL semantics.

\section{The basic \while{} language}
\label{sec:while}

\subsection{Syntax and operational semantics}

We begin by defining the set of arithmetic expressions.
\begin{grammar}
  <expr> ::= \texttt{lit} $\mathbb{N}$ | \texttt{var} $\mathbb{N}$ | <expr>
  <bin> <expr> | <un> <expr>
\end{grammar}

The constructors are respectively literals, a dereference operator \texttt{var},
binary arithmetic operations as well as unary operations. We let $S$ be the set
of lists of natural numbers. The role of $S$ is that of a run-time store whose
entries are referred by their index on the list using constructor \texttt{var}.
We define function $\texttt{eval} : S \times E \to \mathbb{N}$ inductively on
the structure of expressions.

\begin{definition}[Evaluation of expressions in \while]
  \label{def:eval1}
  \begin{align*}
    & \mathtt{eval}~store~(\mathtt{lit}~n) = n \\
    & \mathtt{eval}~store~(\mathtt{var}~l) = \mathtt{get}~store~l\\
    & \mathtt{eval}~store~(e_{1}~b~e_{2}) = (\mathtt{eval}~store~e_{1})~[[b]]~(\mathtt{eval}~store~e_{2})\\
    & \mathtt{eval}~store~(u~e) = [[u]]~(\mathtt{eval}~store~e)
  \end{align*}
\end{definition}

Programs in \while{} language are generated by the following grammar:

\begin{grammar}
  <prog> ::= \texttt{skip} | $\mathbb{N}$ \texttt{:=} <expr> | <prog> ; <prog> |
  \texttt{while} <expr> <prog>
\end{grammar}

The operational semantics of our \while{} language are introduced
in~\Cref{fig:while1}. We are using the notation $\rets{s, x}{s\pr}$ to
denote that program $x$, when supplied with $s : S$, terminates producing store $s\pr$.
Similarly, $\goes{s, x}{s\pr, x\pr}$ means that program $x$, supplied with $s$,
evaluates to $x\pr$ and produces new store $s\pr$.

\begin{figure}
  \begin{framed}
    \framegather
    \begin{gather*}
      \inference{}{\rets{s, \texttt{skip}}{s}} \quad
      \inference{}
      {\rets{s, l~\texttt{:=}~e}{\texttt{update}~s~l~(\texttt{eval}~s~e)}} \quad
      \inference{\rets{s, p}{s\pr}}{\goes{s, p \texttt{;} q}{s\pr , q}} \\
      \inference{\goes{s, p}{s\pr, p\pr}}{\goes{s, p \texttt{;} q}{s\pr , p\pr
          \texttt{;} q}} \quad
      \inference{\texttt{eval}~s~e = 0}{\goes{s, \texttt{while}~
          e~p}{s ,
          \texttt{skip}}} \quad
      \inference{\texttt{eval}~s~e \neq 0}{\goes{s, \texttt{while}~
          e~p}{s , p ; \texttt{while}~e~p}}
    \end{gather*}
    \caption{Semantics of the \while{} language.}
    \label{fig:while1}
  \end{framed}
\end{figure}
\vspace*{-1.38cm} %

\subsection{\while{}, categorically}
The categorical representation of operational semantics has various forms of
incremental complexity but for our purposes we only need to use the most
important one, that of \emph{GSOS laws}~\cite{DBLP:conf/lics/TuriP97}.

\begin{definition}
Given a syntax functor $\Sigma$ and a behavior functor $B$, a GSOS law of
$\Sigma$ over $B$ is a natural transformation $\rho : \Sigma (\Id \times B) \nat
B \Sigma^{*}$, where $(\Sigma^{*}, \eta, \mu)$ is the monad freely generated by
$\Sigma$.
\end{definition}
\begin{example} \label{ex:while-functors}
	Let $E$ be the set of expressions of the \while{}-language.
	Then the syntax functor $\Sigma : \set \to \set$ for \while{} is given by $\Sigma X = \top \uplus (\mathbb N \times E) \uplus (X \times X) \uplus (E \times X)$ where $\uplus$ denotes a disjoint (tagged) union.
	The elements could be denoted as $\mathtt{skip}$, $l := e$, $x_1 ; x_2$ and $\mathtt{while}~e~x$ respectively.	
	The free monad $\Sigma^{*}$ satisfies $\Sigma^{*} X \cong X \uplus \Sigma \Sigma^{*} X$, i.e. its elements are programs featuring program variables from $X$.
	Since \while{}-programs run in interaction with a store and can terminate, the behavior functor is $BX = S \to (S \times \mathrm{Maybe}~X)$, where $S$ is the set of lists of natural numbers and $X \to Y$ denotes the exponential object (internal Hom) $Y^X$.
	
	The GSOS specification of \while{} determines $\rho$. A premise $\goes{s, p}{s', p'}$ denotes an element $(p, b) \in (\Id \times B) X$ where $b(s) = (s', \mathrm{just}~p')$, and a premise $\rets{s, p}{s'}$ denotes an element $(p, b)$ where $b(s) = (s', \mathrm{nothing})$.
	A conclusion $\goes{s, p}{s', p'}$ (where $p \in \Sigma X$ is further decorated above the line to $\bar p \in \Sigma (\Id \times B) X$) specifies that $\rho(\bar p) \in B\Sigma^{*}X$ sends $s$ to $(s', \mathrm{just}~p')$, whereas a conclusion $\rets{s, p}{s'}$ specifies that $s$ is sent to $(s', \mathrm{nothing})$.
	Concretely, $\rho_X : \Sigma(X \times BX) \to B \Sigma^{*} X$ is the function (partially from \cite{DBLP:conf/ctcs/Turi97}):
	\begin{equation*}
		\begin{array}{l c l}
			\mathtt{skip} &\mapsto& \lambda~s.(s, \mathrm{nothing}) \\
    			l~\mathtt{:=}~e &\mapsto& \lambda~s.(\mathtt{update}~s~l~(\mathtt{eval}~s~e), \mathrm{nothing}) \\
    			\mathtt{while}~e~(x,f) &\mapsto& \lambda~s.
				\begin{cases}
					(s, \mathrm{just}~(x~\mathtt{;}~\mathtt{while}~e~x)) & \mathrm{if}~\mathtt{eval}~s~e \neq 0 \\
					(s, \mathrm{just}~(\mathtt{skip})) & \mathrm{if}~\mathtt{eval}~s~e = 0
				\end{cases} \\
			(x,f)~\mathtt{;}~(y,g) &\mapsto& \lambda~s.
				\begin{cases}
					(s\pr, \mathrm{just}~(x\pr~\mathtt{;}~y)) & \mathrm{if}~f(s) = (s\pr, \mathrm{just}~x\pr) \\
					(s\pr, \mathrm{just}~y) & \mathrm{if}~f(s) = (s\pr, \mathrm{nothing}) 
				\end{cases}
			\end{array}
	\end{equation*}
\end{example}

It has been shown by Lenisa et al.~\cite{lenisa-power-watanabe} that there is a
one-to-one correspondence between GSOS laws of $\Sigma$ over $B$ and
\emph{distributive laws} of the free monad $\Sigma^{*}$ over the cofree copointed
endofunctor~\cite{lenisa-power-watanabe} $\Id \times B$.\footnote{A copointed endofunctor is an endofunctor $F$ equipped with a natural transformation $F \nat \Id$.}

\begin{definition}[In~\cite{DBLP:journals/tcs/Klin11}]
A distributive law of a monad $(T,\eta,\mu)$ over a copointed functor
$(H,\epsilon)$ is a natural transformation $\lambda : TH \nat HT$ subject to the
following laws: $\lambda \circ \eta = H\eta$, $\epsilon \circ \lambda =
T\epsilon$ and $\lambda \circ \mu = H\mu \circ \lambda \circ T\lambda$.
\end{definition}

Given any GSOS law, it is straightforward to obtain the corresponding
distributive law via structural induction~(In
\cite{DBLP:journals/entcs/Watanabe02}, prop. 2.7 and 2.8).
By convention, we shall be using the notation $\rho$ for GSOS laws and
$\rho^{*}$ for the equivalent distributive laws unless stated otherwise.

A distributive law $\lambda$ based on a GSOS law $\rho$ gives a category
$\lambda$-Bialg of $\lambda$-bialgebras~\cite{DBLP:conf/lics/TuriP97}, which are
pairs $\Sigma X \xrightarrow{h} X \xrightarrow{k} B X$ subject to the pentagonal law
$ k \circ h = Bh^{*} \circ \rho_{X} \circ \Sigma[id,k]$, where $h^{*}$ is the
inductive extension of $h$. Morphisms in $\lambda$-Bialg are arrows $X 
\rightarrow Y$ that are both algebra and coalgebra homomorphisms at the same
time. The trivial initial B-coalgebra $\init \rightarrow B\init$ lifts uniquely to the
initial $\lambda$-bialgebra $\Sigma \Sigma^{*}\init \xrightarrow{a} \Sigma^{*}\init
\xrightarrow{h_{\lambda}} B \Sigma^{*}\init$, while the trivial final
$\Sigma$-algebra $\Sigma \term \rightarrow \term$ lifts uniquely to the final
$\lambda$-bialgebra $\Sigma B^{\infty}\term \xrightarrow{g_{\lambda}} B^{\infty}\term
\xrightarrow{z} B B^{\infty}\term$~\footnote{We write $B^{\infty}$ for the
 cofree comonad over B, which satisfies $B^\infty X \cong X \times B B^\infty X$.}. Since $\Sigma^{*}\init$ is the set of
programs generated by $\Sigma$ and $B^{\infty}\term$ the set of behaviors
cofreely generated by $B$, the unique bialgebra morphism $f : \Sigma^{*}\init
\rightarrow B^{\infty}\term$ is the \emph{interpretation function} induced by
$\rho$.

\begin{remark}
  \label{rem:mod}
  We write $A$ for $\Sigma^{*}\init$ and $Z$ for
  $B^{\infty}\term$, and refer to $h_{\lambda} : A \to BA$ as the \textit{operational
    model} for $\lambda$ and to $g_{\lambda} : \Sigma Z \to Z$ as the \textit{denotational
    model}~\cite{DBLP:conf/lics/TuriP97}. 
    Note also that $a : \Sigma A \cong A$ and $z : Z \cong BZ$ are invertible.
\end{remark}
\begin{example} \label{ex:while-bialgebras}
	Continuing \cref{ex:while-functors}, the initial bialgebra $A$ is just the set of all \while{}-programs. Meanwhile, the final bialgebra $Z$, which has the meaning of the set of behaviors, satisfies $Z \cong (S \to S \times \mathrm{Maybe}~Z)$.
	In other words, our attacker model is that of an attacker who can count execution steps and moreover, between any two steps, read out and modify the state.
	In \cref{sec:proscons}, we discuss how we hope to consider weaker attackers in the future.
\end{example}

\section{An extra register (Part I)}
\label{sec:whilep}

Let us consider the scenario where a malicious party can observe \emph{more}
information about the execution state of a program, either because information
is being leaked to the environment or the programs are run by a more powerful
machine. A typical example is the  presence of an extra \emph{flags} register that logs the result of a
computation~\cite{DBLP:conf/aplas/PatrignaniCP13,DBLP:conf/csfw/AgtenSJP12,DBLP:conf/csfw/PatrignaniDP16}.
This is the intuition behind the augmented version of \while{} with additional
observational capabilities, \whilep{}.

The main difference is in the behavior so the notation for
transitions has to slightly change. The two main transition types, $\retsv{s,
  \texttt{x}}{s\pr}{v}$ and $\goesv{s, x}{s\pr, x\pr}{v}$ work similarly to \while{}
except for the label $v : \mathbb{N}$ produced when evaluating expressions.
We also allow language terms to interact with the labels by introducing
the constructor $\texttt{obs}~\mathbb{N}~\langle prog \rangle$. When terms evaluate
inside an $\texttt{obs}$ block, the labels are sequentially placed in the
run-time store. The rest of the constructors are identical but the distinction
between the two languages should be clear.

While the expressions are the same as before, the syntax functor is now $\Sigma_\flagsym X = \Sigma X \uplus \mathbb N \times X$,
and the behavior functor is $B_{\flagsym} = S \to \mathbb{N} \times S \times \mathrm{Maybe}~X$.
The full semantics can be found in~\Cref{fig:whilep}. As for \while{}, they
specify a GSOS law $\rho_{\flagsym} : \Sigma_{\flagsym} (\Id \times B_{\flagsym}) \nat B_{\flagsym}
\Sigma_{\flagsym}^{*}$.

\vspace*{-0.3cm} %
\begin{figure}
  \begin{framed}
    \framegather
    \begin{gather*}
      \inference{}{\retsv{s, \texttt{skip}}{s}{0}} \quad
      \inference{v = \texttt{eval}~s~e}
      {\retsv{s, l~\texttt{:=}~e}{\texttt{update}~s~l~v}{v}}  \quad
      \inference{v = \texttt{eval}~s~e & v \neq 0}{\goesv{s, \texttt{while}~
      e~p}{s , \texttt{skip}}{v}} \\
      \inference{\retsv{s, p}{s\pr}{v}}{\goesv{s, p ; q}{s\pr , q}{v}}  \quad
      \inference{\retsv{s, p}{s\pr}{v} & s\prpr =
        \texttt{update}~s\pr~n~v}{\goesv{s, \texttt{obs}~n~p}{s\prpr ,
          \texttt{skip}}{v}} \quad
      \inference{\goesv{s, p}{s\pr, p\pr}{v}}{\goesv{s, p ; q}{s\pr , p\pr
          ; q }{v}} \\
  \inference{\goesv{s, p}{s\pr, p\pr}{v} & s\prpr =
        \texttt{update}~s\pr~n~v}{\goesv{s, \texttt{obs}~n~p}{s\prpr ,
          \texttt{obs}~(n + 1)~p{\pr}}{v}}
       \quad
      \inference{v = \texttt{eval}~s~e & v = 0}{\goesv{s, \texttt{while}~
      e~p}{s , p ; \texttt{while}~e~p}{v}}
    \end{gather*}
    \caption{Semantics of \whilep{}.}
    \label{fig:whilep}
  \end{framed}
\end{figure}
\vspace*{-0.4cm} %

Traditionally, the (in)security of a compiler has been a matter of \emph{full
  abstraction}; a compiler is fully abstract if it
preserves and reflects Morris-style~\cite{morris} contextual equivalence. For our threat
model, where the attacker can directly observe labels, it makes sense to define
contextual equivalence in \whilep~as:
\begin{definition}
\label{ctxeq}
  $p \cong_{\flagsym} q <=> \forall c : C_\flagsym.~c\plug{p}\Downarrow <=> c\plug{q}\Downarrow$
\end{definition}
Where $C$ is the set of one-hole contexts, $\plug{\_} : C_\flagsym \times A_\flagsym -> A_\flagsym$ denotes the plugging function and we write
$p\Downarrow$ when $p$ eventually terminates. Contextual equivalence for
\while{} is defined analogously. It is easy to show that the simple
 ``embedding'' compiler from \while{} to \whilep{} is not fully abstract by
 examining terms $a \triangleq
 \texttt{while}~(\texttt{var}[0])~(0~\texttt{:=}~0)$ and $b \triangleq
 \texttt{while}~(\texttt{var}[0] * 2)~(0~\texttt{:=}~0)$, for which $a \cong b$
 but $a_{\flagsym} \ncong_{\flagsym} b_{\flagsym}$. A context $c \triangleq (\mathtt{obs}~1~\_) ;
 \mathtt{while}~(\mathtt{var[1]}-1)~\mathtt{skip}$ will log the result of the
 $\mathtt{while}$ condition in $a_{\flagsym}$ and $b_{\flagsym}$ in $\mathtt{var}[1]$ and then either
 diverge or terminate depending on the value of $\mathtt{var}[1]$. An
 initial $\mathtt{var}[0]$ value of 1 will cause $c\plug{a}$ to terminate but $c\plug{b}$ to
 diverge.

\subsubsection{Securely extending \whilep{}}

To deter malicious contexts from exploiting the extra information, we introduce
\emph{sandboxing} primitives to help hide it. We add an additional constructor
in \whilep{}, $\lbag \langle progr \rangle \rbag$, and the following inference rules to form
the secure version \whiles{} of \whilep{}.
\begin{align*}
  & \inference{\retsv{s, p}{s\pr}{v}}{\retsv{s, \lbag p \rbag}{s\pr}{0}} \qquad
    \inference{\goesv{s, p}{s\pr, p\pr}{v}}{\goesv{s, \lbag p \rbag}{s\pr , \lbag p\pr \rbag}{0}}
\end{align*}

We now consider the compiler from \while{} to \whiles{} which, along with the
obvious embedding, wraps the the translated terms in sandboxes. This looks to be
effective as programs $a$ and $b$ are now contextually equivalent and the
extra information is adequately hidden. We will show that this compiler is
indeed a \emph{map of distributive laws} between \while{} and \whiles{} but to do so we
need a brief introduction on the underlying theory.

\section{Secure compilers, categorically}

\subsection{Maps of distributive laws}
\label{sec:ctxeq}

Assume two GSOS laws $\rho_{1} : \Sigma_{1} (\Id \times B_{1}) \nat B_{1} \Sigma_{1}^{*}$ and
$\rho_{2} : \Sigma_{2} (\Id \times B_{2}) \nat B_{2} \Sigma_{2}^{*}$, where $(\Sigma_{1}^{*}, \eta_{1},
\mu_{1})$ and $(\Sigma_{2}^{*}, \eta_{2}, \mu_{2})$ are the monads freely generated by
$\Sigma_{1}$ and $\Sigma_{2}$ respectively. We shall regard
pairs of natural transformations $(\sigma : \Sigma_{1}^{*} \nat \Sigma_{2}^{*}, b : B_{1}
\nat B_{2})$ as compilers between the two semantics, where $\sigma$ acts as a
syntactic translation and $b$ as a translation between behaviors.

\begin{remark}
If $A_{1}$ and $A_{2}$ are the sets of terms freely generated by $\Sigma_{1}$ and
$\Sigma_{2}$, we can get the compiler $c : A_{1} -> A_{2}$ from $\sigma$. On the other hand, $b$
generates a function $d : Z_1 \to Z_2$ between behaviors via finality. 
\end{remark}

\begin{remark}
  We shall be writing $B^{c}$ for the cofree copointed endofunctor $\Id \times B$ over $B$ and
  $b^{c} : B^{c}_{1} \nat B^{c}_{2}$ for $\id \times b$.
\end{remark}

\begin{definition}[Adapted from~\cite{DBLP:journals/entcs/Watanabe02}]
  A map of GSOS laws from $\rho_{1}$ to $\rho_{2}$ consists of a natural
  transformation $\sigma : \Sigma_{1}^{*} \nat \Sigma_{2}^{*}$ subject to the monad
  laws $\sigma \circ \eta_{1} = \eta_{2}$ and $\sigma \circ \mu_{1} = \mu_{2}
  \circ \Sigma_{2}^{*}\sigma  \circ \sigma$ paired with a natural transformation
  $b : B_{1} \nat B_{2}$ that satisfies the following \emph{coherence criterion}:
  \begin{center}
    \adjustbox{scale=0.92,center}{
    \begin{tikzcd}
  		\Sigma_1^{*}B^{c}_{1}
  			\arrow[r, "\rho_{1}^{*}"', Rightarrow]
  			\arrow[d, "\sigma~\circ~\Sigma_1^{*}b^{c}"', Rightarrow]
  		& B^{c}_{1}\Sigma_1^{*} \arrow[d, "~b^{c}~\circ~B^{c}_{1}\sigma", Rightarrow]
  		\\
  		\Sigma_2^{*}B^{c}_{2} \arrow[r, "\rho_{2}^{*}", Rightarrow]
  		& B^{c}_{2}\Sigma_2^{*}
  	\end{tikzcd}}
  \end{center}
\end{definition}

\begin{remark} \label{rem:monad}
  \label{rem:fr}
  A natural transformation $\sigma : \Sigma_{1}^{*} \nat \Sigma_{2}^{*}$ subject to the monad
  laws is equivalent to a natural transformation $t : \Sigma_{1} \nat
  \Sigma_{2}^{*}.$
\end{remark}
\begin{theorem}
	If $\sigma$ and $b$ constitute a map of GSOS laws, then we get a compiler $c :
  A_1 \to A_2$ and behavior transformation $d : Z_1 \to Z_2$ satisfying $d \circ
  f_1 = f_2 \circ c : A_1 \to Z_2$. As bisimilarity is exactly equality in the
  final coalgebra (i.e. equality under $f_i : A_{i} \to Z_{i}$), $c$ preserves
  bisimilarity~\cite{DBLP:journals/entcs/Watanabe02}. If $d$ is a monomorphism
  (which, under mild conditions, is the case in $\set$ if every component of $b$
  is a monomorphism), then $c$ also reflects bisimilarity.
\end{theorem}

What is very important though, is that the well-behavedness properties of the
two GSOS laws bestow \emph{contextual} meaning to bisimilarity. Recall that the
gold standard for secure compilation is contextual equivalence (\Cref{ctxeq}),
which is precisely what is observable through program contexts.
Bisimilarity is generally not the same as contextual equivalence, but we can
instead show that in the case of GSOS laws or other forms of distributive laws,
bisimilarity defines the \emph{upper bound} (most fine-grained distinction) of
observability up to program contexts. We shall do so abstractly in the next subsections.

\subsection{Abstract program contexts}

The informal notion of a \emph{context} in a programming language is that of a
program with a hole~\cite{morris}. Thus 
contexts are a syntactic construct that models external interactions with a
program: a single context is an experiment whose 
outcome is the evaluation of the subject program \emph{plugged} in the context.

Na\"ively, one may hope to represent contexts by a functor $H$ sending a set of variables $X$ to the set $HX$ of terms in $\Sigma X$ that may have holes in them.
A complication is that contexts may have holes at any depth (i.e. any number of
operators may have been applied to a hole), whereas $\Sigma X$ is the set of
terms that have exactly one operator in them, immediately applied to variables.
One solution is to think of $Y$ in $HY$ as a set of variables that do
not stand for terms, but for contexts. This approach is fine for multi-hole contexts, but if we also want to consider single-hole contexts and a given single-hole context $c$ is not the hole itself, then precisely one variable in $c$ should stand for a single-hole context, and all other variables should stand for terms.
Thus, in order to support both single- and multi-hole contexts, we make $H$ a two-argument functor, where $H_X Y$ is the set of contexts with term variables from $X$ and context variables from $Y$.

\begin{definition}
  \label{ctxF}
  Let $\cat C$ be a distributive category~\cite{DBLP:journals/mscs/Cockett93} with products $\times$,
  coproducts $\uplus$, initial object $\init$ and terminal object $\term$, as is the case for
  $\set$. A \emph{context functor} for a syntax functor $\Sigma : \mathbb{C} -> \mathbb{C}$, is a functor $H : \mathbb{C} \times \mathbb{C} -> \mathbb{C}$ (with application to $(X, Y)$ denoted as $H_X Y$) such that there
  exist natural transformations $\hole: \forall (X, Y) . \term \to H_X Y$ and
  $\con : \forall X . X \times H_X X \to X \uplus \Sigma X$ making the following
  diagram commute for all $X$:
  \begin{center}
  	\begin{tikzcd}
  		X \times \term
      \arrow{r}{\Proj_1}[swap]{\cong}
      \arrow{d}[swap]{\id_X \times \hole_{(X, X)}}
  		& X \arrow[d, "\Inj_1"]
  		\\
  		X \times H_X X \arrow[r, "\con_X"]
  		& X \uplus \Sigma X
    \end{tikzcd}
  \end{center}
\end{definition}
The idea of the transformation $\con$ is the following: it takes as input a variable $x \in X$ to be plugged into the hole, and a context $c \in H_X X$ with one layer of syntax. The functor $H_X$ is applied again to $X$ rather than $Y$ because $x$ is assumed to have been recursively plugged into the context placeholders $y \in Y$ already. We then make a case distinction: if $c$ is the hole itself, then $\Inj_1~x$ is returned. Otherwise, $\Inj_2~c$ is returned.

\begin{definition} \label{def:plug}
Let $\mathbb C$ be a category as in \cref{ctxF} and assume a syntax functor $\Sigma$ with context functor $H$.
If $\Sigma$ has an initial algebra $(A,q_A)$ (the set of programs) and $H_A$ has a strong initial algebra $(C_A,q_{C_A})$~\cite{DBLP:journals/fuin/Jacobs95} (the set of contexts), then we define the plugging function
$[\![~]\!] : A \times C_A -> A$ as the ``strong inductive extension''~\cite{DBLP:journals/fuin/Jacobs95}
of the algebra structure $[\id_A,q_A] \circ \con_{A} : A \times H_A A \to A$ on $A$, i.e. as the unique morphism that makes the following diagram commute:

\begin{center}
  \begin{tikzcd}
    A \times H_A C_A
    \arrow{rr}{id \times q_{C_{A}}}[swap]{\cong}
    \arrow[d, "{(\pi,st)}"]
    & & A \times C_A
    \arrow[dd, "{[\![~]\!]}"]\\
    A \times H_A(A \times C_A)
    \arrow[d, "{id \times H_A[\![~]\!]}"]\\
    A \times H_A A
    \arrow[r, "\con_{A}"]
    & A \uplus \Sigma A
    & A
    \arrow[from=l, "{[\id,q_A]}"] \\
  \end{tikzcd}
\end{center}
\end{definition}

The above definition of contextual functors is satisfied by both single-hole and
multi-hole contexts, the construction of which we discuss below.

\subsubsection{Multi-hole contexts}
Given a syntax functor $\Sigma$, its multi-hole context functor is simply $H_X Y =
\term \uplus \Sigma Y$. The contextual natural transformation $\con$ is the obvious
map that returns the pluggee if the given context is a hole, and otherwise the context itself (which is then a program):
  \begin{align*}
    & \con : \forall X . X \times (\term \uplus \Sigma X) \to X \uplus \Sigma X \\
    & \con \circ (\id \times i_1) = i_1 \circ \pi_1 : \forall X . X \times \term \to X \uplus \Sigma X \\
    & \con \circ (\id \times i_2) = i_2 \circ \pi_2 : \forall X . X \times \Sigma X \to X \uplus \Sigma X
 \end{align*}
The `pattern matching' is justified by distributivity of $\cat C$.
For $\hole = i_1 : \top \to \top \uplus \Sigma X$, we can see that $\con \circ (\id \times \hole) = i_1 \circ \pi_1$ as required by the definition of a context functor.

\subsubsection{Single-hole contexts}
\label{sec:contexts}

It was observed by McBride~\cite{Mcbride01thederivative} that for inductive
types, i.e. least fixpoints / initial algebras $\mu F$ of certain endofunctors $F$ called
\emph{containers}~\cite{DBLP:journals/tcs/AbbottAG05} or simply \emph{polynomials}, their single-hole contexts are
lists of $\partial F(\mu F)$ where $\partial F$ is the \emph{derivative} of $F$~\footnote{The \emph{list} operator itself arises from the derivative of the free monad operator.}.
Derivatives for containers, which were developed by Abbott et
al. in \cite{DBLP:journals/fuin/AbbottAMG05}, enable us to give a
categorical interpretation of single-hole contexts as long as the syntax functor
$\Sigma$ is a container. 

It would be cumbersome to lay down the entire theory of containers and their
derivatives, so we shall instead focus on the more restricted set of \emph{Simple Polynomial
  Functors}~\cite{DBLP:books/cu/J2016} (or SPF), used to model both syntax and
behavior. Crucially, SPF's are differentiable and
hence compatible with McBride's construction.

\STin{SPF lacks fixed points, meaning that one can't get the derivative of
  ``List X'', as List is a fixed point. Their presence is largely
  inconsequential and only introduces two weird cases when plugging in. Waste of
space if you ask me.}

\begin{definition}[Simple Polynomial Functors]
    The collection of SPF is the least set of functors $\mathbb{C} ->
    \mathbb{C}$ satisfying the following rules:
        \begin{gather*}
      \inference[id]{}{\Id \in \mathrm{SPF}} \quad
      \inference[const]{J \in \mathrm{Obj}(\mathbb{C})}
      {K_{J} \in \mathrm{SPF}} \quad
      \inference[prod]{F,G \in \mathrm{SPF}}{F \times G \in \mathrm{SPF}} \\
      \inference[coprod]{F,G \in \mathrm{SPF}}{F \uplus G \in \mathrm{SPF}} \quad
      \inference[comp]{F,G \in \mathrm{SPF}}{F \circ G \in \mathrm{SPF}}
    \end{gather*}
\end{definition}

We can now define the differentiation action $\partial : \mathrm{SPF} ->
\mathrm{SPF}$ by structural induction. Interestingly, it resembles simple
derivatives for polynomial functions. 

\begin{definition}[SPF derivation rules]
  \begin{equation*}
  	\partial \Id = \top, \quad
  	\partial K_J = \bot, \quad
  	\partial (G \uplus H) = \partial G \uplus \partial H,
  \end{equation*}
  \begin{equation*}
  	\partial (G \times H) = (\partial G \times H) \uplus (G \times \partial H), \quad
  	\partial (G \circ H) = (\partial G \circ H) \times \partial H.
  \end{equation*}
\end{definition}

\begin{example}
The definition of $\con$ for single-hole contexts might look a bit cryptic at
first sight so we shall use a small example from~\cite{Mcbride01thederivative}
to shed some light. In the case of binary trees, locating a hole in a
context can be thought of as traversing through a series of nodes,
choosing left or right according to the placement of the hole until it is found.
At the same time a record of the trees at the non-chosen branches must be kept so that in the  
end the entire structure can be reproduced.

Now, considering that the set of binary trees is the least fixed point of
functor $\mathtt{\term} \uplus (\Id \times \Id)$, then the type of ``abstract''
choice at each intersection is the functor $K_{\mathrm{Bool}} \times \Id$, where $K_{\mathrm{Bool}}$ stands for
a choice of left or right and the $\Id$ part represents the passed structure.
Lists of $(K_{\mathrm{Bool}} \times \Id)~\mathrm{BinTree}$ are exactly the sort of record we need to keep, i.e. they contain the same information as a tree with a single hole.
And indeed $K_{\mathrm{Bool}} \times \Id$ is (up to natural isomorphism) the derivative of $\mathtt{\term} \uplus (\Id
\times \Id)$!
\end{example}

Using derivatives we can define the context functor $H_X Y = \top \uplus
((\partial \Sigma~X) \times Y)$ for syntax functor $\Sigma$. Then the initial algebra
$C_A$ of $H_A$ is indeed $\mathrm{List}~((\partial \Sigma)~A)$, the set of
single-hole contexts for $A \cong \Sigma A$.

\paragraph{Plugging}
Before defining $\con$, we define an auxiliary function $\mathrm{conStep} : \partial \Sigma \times \Id \nat \Sigma$. We defer the reader
to~\cite{Mcbride01thederivative} for the full definition of $\mathrm{conStep}$,
which is inductive on the SPF $\Sigma$, and shall instead only define the case for coproducts. So, for $\partial (F
\uplus G) = \partial F \uplus \partial G$ we have:
\begin{align*}
  & \mathrm{conStep}_{F \uplus G} : (\partial F \uplus \partial G) \times \Id \nat F \uplus G \\
  & \mathrm{conStep}_{F \uplus G} \circ (\Inj_1 \times \id) = \Inj_1 \circ \mathrm{conStep}_F : \partial F \times \Id \nat F \uplus G \\
  & \mathrm{conStep}_{F \uplus G} \circ (\Inj_2 \times \id) = \Inj_2 \circ \mathrm{conStep}_G : \partial G \times \Id \nat F \uplus G  
\end{align*}
We may now define $\con : X \times H_X X \to X \uplus \Sigma X$ as follows:
  \begin{align*}
    & \con : \forall X . X \times (\term \uplus (\partial \Sigma~X \times X)) \to X \uplus \Sigma X \\
    & \con \circ (\id \times i_1) = i_1 \circ \pi_1 : \forall X . X \times \term \to X \uplus \Sigma X \\
    & \con \circ (\id \times i_2) = i_2 \circ \mathrm{conStep}_{\Sigma} \circ \pi_2 : \forall X . X \times (\partial \Sigma~X \times X) \to X \uplus \Sigma X
 \end{align*}
By setting $\hole = i_1 : \top \to \top \uplus (\partial \Sigma~X \times X)$ we can see that $\con \circ (\id \times \hole) = i_1 \circ \pi_1$
as required by~\Cref{ctxF}.

\STin{I bet \textbackslash Id is used on purpose. \AN{I dislike words put in mathmode just like that, but if you don't, then you can always redefine my commands to be just the words.}}

\subsection{Contextual coclosure}
\label{sec:proof}

Having established a categorical notion of contexts, we can now move
towards formulating contextual categorical arguments about bisimilarity. We
assume a context functor $H$ for $\Sigma$ such that $H_A$ has strong initial algebra
$(C_A, q_{C_A})$ (the object containing all contexts).

First, since we prefer to work in more general categories than just $\Set$, we
will encode relations $R \subseteq X \times Y$ as spans $X \xleftarrow{r_1} R
\xrightarrow{r_2} Y$. One may wish to consider only spans for which $(r_1, r_2)
: R \to X \times Y$ is a monomorphism, though this is not necessary for our
purposes.

We want to reason about contextually closed relations on the set of terms $A$,
which are relations such that $a_1 \mathrel{R} a_2$ implies $(c \plug{a_1})
\mathrel{R} (c \plug{a_2})$ for all contexts $c \in C_A$. Contextual equivalence
will typically be defined as the co-closure of equitermination: the greatest
contextually closed relation that implies equitermination. For spans, this
becomes:
\begin{definition}
	In a category as in \cref{ctxF},
	a span $A \xleftarrow{r_1} R \xrightarrow{r_2} A$ is called contextually closed if there is a morphism $\plug{~} : C_A \times R \to R$ making the following diagram commute:
	\begin{center}
		\begin{tikzcd}
			C_A \times A
				\arrow{d}{\plug{~}}
			& C_A \times R
				\arrow{l}[swap]{\id \times r_1}
				\arrow{r}{\id \times r_2}
				\arrow{d}{\plug{~}}
			& C_A \times A
				\arrow{d}{\plug{~}}
			\\
			A
			& R
				\arrow{l}[swap]{r_1}
				\arrow{r}{r_2}
			& A
		\end{tikzcd}
	\end{center}
	The contextual co-closure $A \xleftarrow{\bar r_1} \bar R \xrightarrow{\bar r_2} A$ of an arbitrary span $A \xleftarrow{r_1} R \xrightarrow{r_2} A$ is the final contextually closed span on $A$ with a span morphism $\bar R \to R$.
\end{definition}
We call terms bisimilar if the operational semantics $f : A \to Z$ assigns them
equal behaviors:
\begin{definition}
  \label{strbis}
	We define (strong) bisimilarity $\bisim$ as the pullback of the equality span $(\id_Z, \id_Z) : Z \to Z \times Z$ along $f \times f : A \times A \to Z \times Z$ (if existent).
\end{definition}

\begin{theorem}
  \label{ctxClosure}
  Under the assumptions of \ref{def:plug}, bisimilarity (if existent) is contextually closed.
\end{theorem}
\begin{proof}
	We need to give a morphism of spans from $C_A \times {\bisimspan}$ to $\bisimspan$:
	\begin{center}
		\begin{tikzcd}
			C_A \times A
				\arrow{d}{\plug{~}}
			& C_A \times \bisimspan
				\arrow{l}[swap]{\id \times r_1}
				\arrow{r}{\id \times r_2}
				\arrow[dotted]{d}
			& C_A \times A
				\arrow{d}{\plug{~}}
			\\
			A
				\arrow{d}{f}
			& \bisimspan
				\arrow{l}[swap]{r_1}
				\arrow{r}{r_2}
				\arrow{d}{w}
			& A
				\arrow{d}{f}
			\\
			Z
			& Z
				 \arrow{l}[swap]{\id_Z}
				 \arrow{r}{\id_Z}
			& Z.
		\end{tikzcd}
	\end{center}
	By definition of $\bisimspan$, it suffices to give a morphism of spans to the
  equality span on $Z$, i.e. to prove that $f \circ \plug{~} \circ (\id \times
  r_1) = f \circ \plug{~} \circ (\id \times r_2)$. To this end, consider the
  following diagram (parameterized by $i \in \{1, 2\}$), in which every polygon
  is easily seen to commute:
  \begin{center}
    \adjustbox{scale=0.9,center}{
		\begin{tikzcd}
			\bisimspan \times H_AC_A
			    \arrow{rrrr}{\id \times q_{C_A}}[swap]{\cong}
				\arrow[bend right=15]{rrd}[swap]{r_i \times \id}
				\arrow{dd}[swap]{(\Proj_1, H_A (r_i \times \id) \circ \strength)}
			&&&& \bisimspan \times C_A \arrow{d}{r_i \times \id}
			\\
			&& A \times H_AC_A
				\arrow{rr}{\id \times q_{C_A}}[swap]{\cong}
				\arrow{d}[swap]{(\Proj_1 , \strength)}
			&& A \times C_A \arrow{dd}{\plug{~}}
			\\
			\bisimspan \times H_A(A \times C_A) \arrow{rr}[swap]{r_i \times \id}
				\arrow{d}[swap]{\id \times  H_A \plug{~}}
			&& A \times H_A(A \times C_A)
				\arrow{d}[swap]{\id \times  H_A \plug{~}}
			\\
			\bisimspan \times H_AA \arrow{rr}[swap]{r_i \times \id}
				\arrow{d}[swap]{\id \times  H_A f}
			&& A \times H_AA \arrow{r}[swap]{\con_A}
				\arrow{dd}[swap]{f \times  H_f f}
			& A \uplus \Sigma A \arrow{r}[swap]{[\id, q_ A]}
				\arrow{dd}[swap]{f \uplus  \Sigma f}
			& A \arrow{dd}{f}
			\\
			\bisimspan \times H_AZ
				\arrow{d}[swap]{\id \times  H_f \id_Z}
			\\
			\bisimspan \times H_ZZ \arrow{rr}[swap]{w \times \id}
			&& Z \times H_ZZ \arrow{r}[swap]{\con_Z}
			& Z \uplus \Sigma Z \arrow{r}[swap]{[\id, q_ Z]}
			& Z
		\end{tikzcd}}
\end{center}
The bottom-right square stems from the underlying GSOS law: it is the algebra
homomorphism part of the bialgebra morphism between the initial and the final
bialgebras. Commutativity of the outer diagram reveals that $f \circ \plug{~}
\circ (r_i \times \id)$ is, regardless of $i$, \emph{the} strong inductive
extension of $[\id_Z, q_Z] \circ \con_Z \circ (w \times H_f \id_Z) : \bisimspan \times H_AZ \to Z$. Thus, it is independent of
$i$. \qed
\end{proof}

\begin{corollary}
  \label{maincor}
  In $\set$, bisimilarity is its own contextual coclosure: $a \bisim b <=> \forall~c \in C_{A}.~c\plug{a} \bisim c\plug{b}$
\end{corollary}
\begin{corollary} \label{thm:bisim-ctxeq}
	In $\set$, bisimilarity implies contextual equivalence.\footnote{Note that we
    can \emph{not} conclude that preservation of bisimilarity would imply preservation of contextual equivalence.}
\end{corollary}
\begin{proof}
	Bisimilarity implies equitermination. This yields an implication between their coclosures. \qed
\end{proof}

Comparing~\Cref{maincor} to contextual equivalence in~\Cref{ctxeq} reveals their key
difference. Contextual equivalence makes minimal assumptions on the underlying
observables, which are simply divergence and termination. On the other hand, the
contextual coclosure of bisimilarity assumes maximum observability (as dictated
by the behavior functor) and in that sense it represents the upper bound of what can be observed through contexts. Consequently, this criterion is useful if the
observables adequately capture the threat model, which is true for the examples
that follow.

This theorem echoes similar results in the broader study of coalgebraic
bisimulation~\cite{Bartels04ongeneralised, DBLP:journals/mscs/RotBBPRS17}. There
are, however, two differences. The first is that our theorem allows for extra
flexibility in the definition of contexts as the theorem is parametric on the
context functor. Second, by making the context construction explicit we can
directly connect (the contextual coclosure of) bisimilarity to contextual
equivalence (\cref{thm:bisim-ctxeq}) and so have a more convincing argument for using maps of
distributive laws as secure compilers.

\section{An extra register (Part II)}
\label{sec:ex1}

The next step is to define the syntax and behavior natural transformations. The
first compiler, $\sigma_{\flagsym} : \Sigma \nat \Sigma_{\flagsym}$, is a very simple 
mapping of constructors in \while{} to their \whilep{} counterparts. The second
natural transformation, $\sigma_{\secflagsym} : \Sigma \nat \Sigma_{\secflagsym}^{*}$, is more complex
as it involves an additional layer of syntax in \whiles{}.

\begin{definition}[Sandboxing natural transformation]
  \label{sandbox}
  Consider the natural transformation $e : \Sigma \nat \Sigma_{\secflagsym}$ which embeds
  $\Sigma$ in $\Sigma_{\secflagsym}$.
	Using PL notation, we define $\sigma_{\secflagsym} : \Sigma X \to \Sigma^{*}_{\secflagsym} X : p \mapsto \lbag e(p) \rbag$.
	This yields a monad morphism $\sigma_{\secflagsym}^{*} : \Sigma^{*} \to \Sigma_{\secflagsym}^{*}$ (\cref{rem:monad}).
\end{definition}

Defining the natural translations between behaviors is a matter of choosing a
designated value for the added observable label. The only constraint is that the
chosen value has to coincide with the label that the sandbox produces. $B_{\flagsym}$ and
$B_{\secflagsym}$ are identical so we need a single natural transformation $b : B \nat B_{\flagsym/\secflagsym}$:
  \begin{align*}
    & b : \forall X.~(S -> S \times \mathrm{Maybe}~X) -> S -> \mathbb{N} \times S \times \mathrm{Maybe}~X \\
    & b~f = \lambda (s : S) -> (0 , f(s))
  \end{align*}

\subsubsection{\while{} to \whilep{}}

We now have the natural translation pairs $(\sigma_{\flagsym} , b)$ and $(\sigma_{\secflagsym} , b)$, which
allows us to check the coherence criterion from~\Cref{sec:ctxeq}.
We shall be using a graphical notation that provides for a good
intuition as to what failure or success of the criterion \emph{means}. For
example, \cref{fig:coh1} shows failure of the coherence criterion for the first pair.

\begin{wrapfigure}{r}{0.52\textwidth}
  \vspace*{-0.3cm}
\begin{tikzpicture}[node distance=2cm,thick,scale=0.9, every node/.style={scale=0.9}]
  \node (A) at (0, 0) {$l~\texttt{:=}~e$};
  \node (B) at (4, 0) {\inference{v = \texttt{eval}~s~e}{\rets{s}{\texttt{update}~s~l~v}}};
  \node (C) at (0, -1.6) {$l~\texttt{:=}~e$};
  \node (D) at (4, -1.6) {\inference{v = \texttt{eval}~s~e}{\retsv{s}{\texttt{update}~s~l~v}{\textcolor{red}{v}/\textcolor{blue}{0}}}};

\draw[->] (A) edge node[above] {$\rho^{*}$} (B);
\draw[->] (A) edge node[right,black] {$\sigma^{*}_{\flagsym} \circ \Sigma^{*} b^{c}$} (C);
\draw[->][blue] (B) edge node[right,black] {$b^{c} \circ B^{c} \sigma^{*}_{\flagsym}$} (D);
\draw[->][red] (C) edge node[above,black] {$\rho^{*}_{\flagsym}$} (D);

\end{tikzpicture}
\vspace*{-0.15cm}
\caption{Failure of the criterion for $(\sigma_{\flagsym} , b)$.}
\label{fig:coh1}
\end{wrapfigure}

The horizontal arrows in the diagram represent the two semantics, $\rho^{*}$ and
$\rho^{*}_{\flagsym}$, while the vertical arrows are the two horizontal compositions of the
natural translation pair. The top-left node holds an element of
$\Sigma^{*} (\Id \times B)$, which in this case is an assignment operation. The two
rightmost nodes represent behaviors, so the syntactic element is missing from
the left side of the transition arrows.

In the upper path, the term is first applied to the GSOS law $\rho^{*}$ and the result is
then passed to the translation pair, thus producing the designated label $0$,
typeset in blue for convenience. In the lower path, the term is first applied to
the translation and then goes through the target semantics, $\rho^{*}_{\flagsym}$, where
the label $v$ is produced. It is easy to find such an $s$ so that
$\textcolor{red}{v} \neq \textcolor{blue}{0}$.

\subsubsection{\while{} to \whiles{}}

\begin{wrapfigure}{r}{0.55\textwidth}
\vspace*{-0.3cm}
\begin{tikzpicture}[node distance=2cm,thick,scale=0.9, every node/.style={scale=0.9}]
  \node (A) at (0, 0) {$l~\texttt{:=}~e$};
  \node (B) at (4, 0) {\inference{v = \texttt{eval}~s~e}{\rets{s}{\texttt{update}~s~l~v}}};
  \node (C) at (0, -1.9) {$\lbag l~\texttt{:=}~e \rbag$};
  \node (D) at (4, -1.9) {\inference{\inference[]{v = \texttt{eval}~s~e}{\textcolor{red}{\retsvr{s, l~\texttt{:=}~e}{\texttt{update}~s~l~v}{v}}}}{\retsv{s}{\texttt{update}~s~l~v}{\textcolor{purple}{0}}}};

\draw[->] (A) edge node[below] {$\rho^{*}$} (B);
\draw[->] (A) edge node[right,black] {$\sigma^{*}_{\secflagsym} \circ \Sigma^{*} b^{c}$} (C);
\draw[->][blue] (B) edge node[right,black] {$b^{c} \circ B^{c} \sigma^{*}_{\secflagsym}$} (D);
\draw[->][red] (C) edge node[above,black] {$\rho_{\secflagsym}^{*}$} (D);

\end{tikzpicture}
\vspace*{-0.2cm}
\caption{The coherence criterion for $(\sigma_{\secflagsym} , b)$.}
\label{fig:coh2}
\end{wrapfigure}

The same example is investigated for the second translation pair $(\sigma_{\secflagsym}
, b)$. \Cref{fig:coh2} shows what happens when we test the same case as before.
Applying $\rho_{\secflagsym}^{*}$ to $\lbag l~\texttt{:=}~e \rbag$ is similar to
$\rho_{\secflagsym}$ acting twice. The innermost transition is the intermediate step and
as it only appears in the bottom path it is typeset in red. This time the
diagram commutes as the label produced in the inner layer, $\textcolor{red}{v}$,
is effectively erased by the sandboxing rules of \whiles{}.

\subsubsection{An endo-compiler for \whiles{}}
If $A_{\secflagsym}$ is the set of closed terms for \whiles{}, the compiler $u : A_{\secflagsym} ->
A_{\secflagsym}$, which ``escapes'' \whiles{} terms from their sandboxes can be
elegantly modeled using category theory. As before, it is not possible to
express it using a simple natural transformation $\Sigma_{\secflagsym} \nat \Sigma_{\secflagsym}$. We
can, however, use the \emph{free pointed
  endofunctor}~\cite{lenisa-power-watanabe} over $\Sigma_{\secflagsym}$,  $\Id 
\uplus \Sigma_{\secflagsym}$. What we want is to map non-sandboxed
terms to themselves and lift the extra layer of syntax from sandboxed terms.
Intuitively, for a set of variables $X$, $\Sigma_{\secflagsym} X$ is one layer of syntax
``populated'' with elements of $X$. If $X \uplus \Sigma_{\secflagsym} 
X$ is the union of $\Sigma_{\secflagsym} X$ with the set of variables $X$, lifting the
sandboxing layer is mapping the $X$ in $\lbag X \rbag$ to the left of $X
\uplus \Sigma_{\secflagsym} X$ and the rest to themselves at the right.

\begin{wrapfigure}{r}{0.5\textwidth}
\begin{tikzpicture}[node distance=2cm,thick,scale=0.9, every node/.style={scale=0.9}]
  \node (A) at (0, 0) {\inference{\goesv{s,p}{s',q}{v}}{\lbag p \rbag}};
  \node (B) at (4, 0) {\goesv{s}{s', \lbag q \rbag}{0}};
  \node (C) at (0, -1.6) {\inference{\goesv{s,p}{s',q}{v}}{p}};
  \node (D) at (4, -1.6) {\goesv{s}{s', q}{\textcolor{red}{v}/\textcolor{blue}{0}}};

\draw[->] (A) edge node[below] {$\rho_{\secflagsym}^{*}$} (B);
\draw[->] (A) edge node[left,black] {$\sigma^{*}_{u}$} (C);
\draw[->][blue] (B) edge node[right,black] {$B^{c} \sigma^{*}_{u}$} (D);
\draw[->][red] (C) edge node[above,black] {$\rho_{\secflagsym}^{*}$} (D);

\end{tikzpicture}
\vspace*{-0.2cm}
\caption{Failure of the criterion for $\sigma_{u}$.}
\vspace*{-0.5cm}
\label{fig:coh3}
\end{wrapfigure}

This is obviously not a secure compiler as it allows discerning previously
indistinguishable programs. As we can see in~\Cref{fig:coh3}, the coherence
criterion fails in the expected manner.

\section{State mismatch}
\label{sec:state}

Having established our categorical foundations, we shall henceforth focus on
examples. The first one involves a compiler where the target machine is not
necessarily more powerful than the source machine, but the target \emph{value}
primitives are not isomorphic to the ones used in the source. This is a
well-documented problem~\cite{DBLP:journals/toplas/PatrignaniAS0CP15}, which has
led to failure of full abstraction before~\cite{DBLP:conf/csfw/AgtenSJP12,
  DBLP:conf/popl/FournetSCDSL13, DBLP:journals/tcs/Kennedy06}.

For example, we can repeat the development of \while{} except we substitute
natural numbers with integers. We call this new version \whilez{}.
\begin{grammar}
  <expr> ::= \texttt{lit} $\mathbb{Z}$ | \texttt{var} $\mathbb{N}$ | <expr>
  <bin> <expr> | <un> <expr>
\end{grammar}
The behavior functor also differs in that the store type $S$ is substituted with
$S_{\intsym}$, the set of lists of integers. We can define the behavioral
natural transformation $b_\intsym : B \nat B_\intsym$ as the best ``approximation'' between the two
behaviors. In $\set$:
  \begin{align*}
    & b_\intsym : \forall X.~(S -> S \times (\term \uplus X)) -> S_\intsym -> S_\intsym \times (\term \uplus X) \\
    & b_\intsym~f = [\text{to}\mathbb{Z} , \id] \circ f \circ \text{to}\mathbb{N}
  \end{align*}

  \BeforeBeginEnvironment{wrapfigure}{\setlength{\intextsep}{2pt}}
  \begin{wrapfigure}{r}{0.5\textwidth}
    \begin{tikzpicture}[node distance=2cm,thick,scale=0.9, every node/.style={scale=0.9}]
      \node (A) at (0, 0) {$\texttt{0 := min(var[0],0)}$};
      \node (B) at (4, 0) {\rets{{[n]}}{{[0]}}};
      \node (C) at (0, -1.4) {$\texttt{0 := min(var[0],0)}$};
      \node (D) at (4, -1.4) {\rets{{[-1]}}{{[\textcolor{red}{-1}/\textcolor{blue}{0}]}}};

      \draw[->] (A) edge node[below] {$\rho^{*}$} (B);
      \draw[->] (A) edge node[left,black] {$\Sigma^{*} b^{c}_{\intsym}$} (C);
      \draw[->][blue] (B) edge node[left,black] {$b^{c}_{\intsym}$} (D);
      \draw[->][red] (C) edge node[above,black] {$\rho^{*}_{\intsym}$} (D);
    \end{tikzpicture}
    \caption{Failure of the criterion for $(\id , b_{\intsym})$.}
    \label{fig:coh4}
  \end{wrapfigure}
  \BeforeBeginEnvironment{wrapfigure}{\setlength{\intextsep}{0pt}}

Where $\text{to}\mathbb{N}$ replaces all negative numbers in the store with 0 and $\text{to}\mathbb{Z}$ typecasts $S$
to $S_{\intsym}$. It is easy to see that the identity compiler from \while{} to
\whilez{} is not fully abstract. For example, the expressions $0$ and $\texttt{min}(\texttt{var}[0],0)$ are identical in
\while{} but can be distinguished in \whilez{} (if $\mathtt{var}[0]$ is negative). This is reflected in the
coherence criterion diagram for the identity compiler in~\Cref{fig:coh4}, when
initiating the store with a negative integer.

\begin{wrapfigure}{r}{0.5\textwidth}
\begin{tikzpicture}[node distance=2cm,thick,scale=0.9, every node/.style={scale=0.9}]
  \node (A) at (0, 0) {$\texttt{0 := min(var[0],0)}$};
  \node (B) at (4, 0) {\rets{{[n]}}{{[0]}}};
  \node (C) at (0, -1.6) {$\langle \texttt{0 := min(var[0],0)} \rangle$};
  \node (D) at (4, -1.6) {\rets{{[-1]}}{{[\textcolor{purple}{0}]}}};

\draw[->] (A) edge node[below] {$\rho^{*}$} (B);
\draw[->] (A) edge node[left,black] {$\sigma^{*}_{\intsym} \circ \Sigma^{*} b^{c}_{\intsym}$} (C);
\draw[->][blue] (B) edge node[left,black] {$b^{c}_{\intsym} \circ B^{c}\sigma^{*}_{\intsym}$} (D);
\draw[->][red] (C) edge node[above,black] {$\rho^{*}_{\intsym}$} (D);

\end{tikzpicture}
\caption{The coherence criterion for $(\sigma_{\intsym}, b_{\intsym})$.}
\label{fig:coh5}
\end{wrapfigure}

The solution is to create a special environment where \whilez{} forgets about
negative integers, in essence copying what $b_{\intsym}$ does on the variable store.
This is a special kind of sandbox, written $\langle\_\rangle$, for which we
introduce the following rules:
\begin{align*}
  & \inference{\rets{\text{to}\mathbb{N}(s), p}{s\pr}}{\rets{s, \langle p \rangle}{s\pr}} \qquad
  \inference{\goes{\text{to}\mathbb{N}(s), p}{s\pr, p\pr}}{\goes{s, \langle p \rangle}{s\pr , \langle p\pr \rangle}}
\end{align*}
We may now repeat the construction from~\Cref{sandbox} to define the compiler
$\sigma_{\intsym}$. We can easily verify that the pair $(\sigma_{\intsym}, b_{\intsym})$
constitutes a map of distributive laws. For instance, \Cref{fig:coh5}
demonstrates how the previous failing case now works under $(\sigma_{\intsym}, b_{\intsym})$.

\section{Control Flow}
\label{sec:control}

Many low-level languages support unrestricted control flow in the
form of jumping or branching to an address. On the other hand, control flow
in high-level languages is usually restricted (think if-statements or function
calls). A compiler from the high-level to the low-level might be insecure as it
exposes source-level programs to illicit control flow. This is another
important and well-documented example of failure of full
abstraction~\cite{DBLP:conf/csfw/AgtenSJP12,
  DBLP:journals/toplas/PatrignaniAS0CP15,
  DBLP:conf/ccs/AbateABEFHLPST18,Patrignani2019FormalAT}.

\begin{figure}
  \begin{framed}
    \framegather
    \begin{gather*}
      \inference{}{\rets{s,0,\texttt{stop}~[\texttt{;;}~x]}{s,0}} \quad
      \inference{v = \texttt{eval}~s~e & s\pr =
        \texttt{update}~s~n~v}{\goes{s,0,\texttt{assign}~n~v~[\texttt{;;}~x]}{s\pr,1,\texttt{assign}~n~v~[\texttt{;;}~x]}}
      \\
      \inference{\texttt{PC} \geq 0 & \rets{s,\texttt{PC},x}{s\pr,\texttt{PC}\pr}}{\rets{s,\texttt{PC}+1,i~
          \texttt{;;}~x}{s\pr,\texttt{PC}\pr + 1}} \quad
      \inference{v = \texttt{eval}~s~e & v =
        0}{\goes{s,0,\texttt{br}~e~z~[\texttt{;;}~x]}{s,1,\texttt{br}~e~z~[\texttt{;;}~x]}}
      \\
      \inference{v = \texttt{eval}~s~e & v \neq
        0}{\goes{s,0,\texttt{br}~e~z~[\texttt{;;}~x]}{s,z,\texttt{br}~e~z~[\texttt{;;}~x]}} \quad
      \inference{\texttt{PC} < 0}{\rets{s,\texttt{PC},i~
          \texttt{;;}~x}{s,\texttt{PC}}}
      \\
      \inference{p = \texttt{nop}~[\texttt{;;}~x]}{\goes{s,0,p}{s,1,p}} \quad
      \inference{\texttt{PC} \geq 0 & \goes{s,\texttt{PC},x}{s\pr,\texttt{PC}\pr,x\pr}}{\goes{s,\texttt{PC}+1,i~
          \texttt{;;}~x}{s\pr,\texttt{PC}\pr + 1,i~\texttt{;;}~x\pr}} \quad
      \inference{\texttt{PC} \neq 0}{\rets{s,\texttt{PC},i}{s,\texttt{PC}}}
    \end{gather*}
    \caption{Semantics of the \low{} language. Elements in square brackets are
      optional.}
    \label{fig:low}
  \end{framed}
\end{figure}

We introduce low-level language \low{}, the programs of which are non-empty lists of
instructions. \low{} differs significantly from \while{} and its derivatives in
both syntax and semantics. For the syntax, we define the set of instructions
$\langle inst \rangle$ and set of programs $\langle asm \rangle$.
\begin{grammar}
  <inst> ::= \texttt{nop} | \texttt{stop} | \texttt{assign} $\mathbb{N}$ <expr>  |
  \texttt{br} <expr> $\mathbb{Z}$

  <asm> ::= <inst> | <inst> \texttt{;;} <asm>
\end{grammar}

Instruction \texttt{nop} is the no-operation, \texttt{stop} halts
execution and \texttt{assign} is analogous to the assignment operation in
\while. The \texttt{br} instruction is what really defines \low{}, as it stands
for bidirectional relative branching.

\subsubsection{(Bialgebraic) Semantics for \low}

\Cref{fig:low} shows the operational semantics of \low{}. The execution state of
a running program consists of a run-time store and the program counter register
$\texttt{PC} \in \mathbb{Z}$ that points at the instruction being processed. If the program
counter is zero, the leftmost instruction is executed. If the program counter is
greater than zero, then the current instruction is further to the right.
Otherwise, the program counter is out-of-bounds and execution stops.
The categorical interpretation suggests a GSOS law $\rho_{L}$ of syntax
functor $\Sigma_{L} X = \text{inst} \uplus (\text{inst} \times
X)$ over behavior functor $B_{L} X = S \times \mathbb{Z} -> S \times \mathbb{Z}
\times \mathrm{Maybe}~X$.

\subsubsection{An insecure compiler}
\vspace*{-0.15cm}

This time we start with the behavioral translation, which is less obvious as we
have to go from $BX = S -> S \times \mathrm{Maybe}~X$ to 
$B_{L}X = S \times \mathbb{Z} -> S \times \mathbb{Z} \times \mathrm{Maybe}~X$. The increased arity in $B_{L}$ poses an interesting question as to what
the program counter should mean in \while{}. It makes sense to consider the program counter
in \while{} as zero since a program in \while{} is treated uniformly as a single
statement.
  \begin{align*}
    & b_{L} : \forall X.~(S -> S \times \mathrm{Maybe}~X) -> S \times \mathbb{Z} -> S \times \mathbb{Z} \times \mathrm{Maybe}~X \\
      & b_{L}~f~(s, 0) =
        \begin{cases}
        (s\pr, 1, \mathrm{nothing}) & \text{if}~f~s = (s\pr, \mathrm{nothing}) \\
        (s\pr, 0, \mathrm{just}~y) & \text{if}~f~s = (s\pr, \mathrm{just}~y)
      \end{cases}
                  \\
      & b_{L}~f~(s, n \neq 0) = (s, n, \mathrm{nothing})
  \end{align*}
When it comes to translating terms, a typical compiler from \while{} to \low{}
would untangle the tree-like structure of \while{} and convert it to a list of
\low{} instructions. For \texttt{while} statements, the compiler would use
branching to simulate looping in the low-level.
\begin{example}
  Let us look at a simple case of a loop. The \while{} program \\
    \mbox{\texttt{while} (\texttt{var} 0 \texttt{<} 2) (1 \texttt{:=}
      \texttt{var} 1 \texttt{+} 1)} is compiled to \\
    \mbox{\texttt{br} \texttt{!}(\texttt{var} 0 \texttt{<} 2) 3 \texttt{;;}
    \texttt{assign} 1 (\texttt{var} 1 \texttt{+} 1) \texttt{;;}
    \texttt{br} (\texttt{lit} 1) -2}
  \label{ex1}
\end{example}

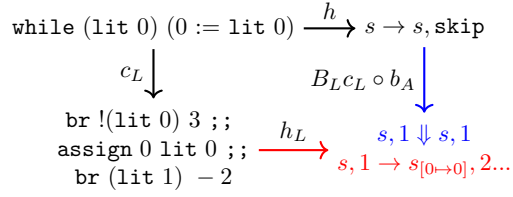
\begin{wrapfigure}{r}{0.57\textwidth}
\begin{tikzpicture}[node distance=2cm,thick,scale=0.9, every node/.style={scale=0.9}]
  \node (A) at (0, 0) {\texttt{while}~(\texttt{lit} 0)~(0 := \texttt{lit} 0)};
  \node (B) at (4, 0) {\goes{s}{s, \texttt{skip}}};
  \node (C) at (0, -1.8) {$\begin{array}{@{}c@{}}\texttt{br}~!(\texttt{lit}~0)~3~\texttt{;;}~\\\texttt{assign}~0~\texttt{lit}~0~\texttt{;;}\\\texttt{br}~(\texttt{lit}~1)~-2 \end{array}$};
    \node (D) at (4, -1.8) {$\begin{array}{@{}c@{}}\textcolor{blue}{\rets{s,1}{s,1}} \\ \textcolor{red}{\goes{s,1}{s_{[0 \mapsto 0]},2...}} \end{array}$};

\draw[->] (A) edge node[above] {$h$} (B);
\draw[->] (A) edge node[left,black] {$c_{L}$} (C);
\draw[->][blue] (B) edge node[left,black] {$B_{L}c_{L} \circ b_{A}$} (D);
\draw[->][red] (C) edge node[above,black] {$h_{L}$} (D);
\end{tikzpicture}
\caption{$c_{L}$ is not a coalgebra homomorphism.}
\label{fig:cohLow}
\end{wrapfigure}

This compiler, called $c_{L}$, cannot be defined in terms of a natural
transformation $\Sigma \nat \Sigma^{*}_{L}$ as per~\Cref{rem:fr}, but it is
inductive on the terms of the source language. In this case we can directly
compare the two operational models $b_{A} \circ h : A -> B_{L}A$ (where $h : A \to BA$) and $h_{L} : A_{L}
-> B_{L} A_{L}$ and notice that $c_{L} : A -> A_{L}$ is not a coalgebra
homomorphism (\Cref{fig:cohLow}). The key is that the program counter in \low{}
allows for finer observations on programs. Take for 
example the case for $\texttt{while}~(\texttt{lit}~0)~(0 := \texttt{lit}~0)$, where the loop is
always skipped. In \low{}, we can still access the loop body by simply
pointing the program counter to it. This is a realistic attack scenario because
\low{} allows manipulation of the program counter via the \texttt{br} instruction.

\subsubsection{Solution}
By comparing the semantics between \while{} in~\Cref{fig:while1} and \low{}
in~\Cref{fig:low} we find major differences. The first one is the reliance of 
\low{} to a program counter which keeps track of execution, whereas \while{}
executes statements from left to right. Second,
the sequencing rule in \while{} dictates that statements are 
removed from the program state\footnote{We are not referring to the
  store, but to the internal, algebraic state.} upon completion. On the
other hand, \low{} keeps the program state intact at all times. Finally, there
is a stark contrast between the two languages in the way they handle
\texttt{while} loops.

To address the above issues we introduce a new sequencing primitive
\texttt{;;}$_{c}$ and a new looping primitive \texttt{loop} for \low{}, which
prohibit illicit control flow and properly propagate the internal state.
Furthermore, we change the semantics of the singleton \texttt{assign}
instruction so that it mirrors the peculiarity of its \while{} counterpart. The
additions  can be found in~\Cref{fig:low+}.

\begin{figure}
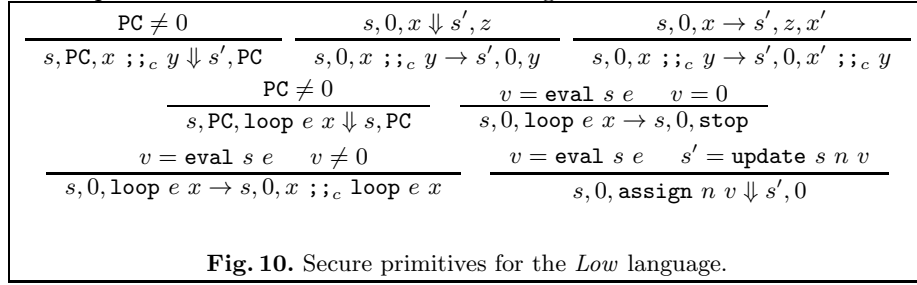

  \begin{framed}
    \framegather
    \begin{gather*}
        \inference{\texttt{PC} \neq 0}{\rets{s, \texttt{PC}, x~\texttt{;;}_{c}~y}{s\pr, \texttt{PC}}} ~
        \inference{\rets{s, 0, x}{s\pr, z}}{\goes{s, 0, x~\texttt{;;}_{c}~y}{s\pr, 0, y}} ~
        \inference{\goes{s, 0, x}{s\pr, z, x\pr}}{\goes{s, 0, x~\texttt{;;}_{c}~y}{s\pr , 0, x\pr~\texttt{;;}_{c}~y}} \\
        \inference{\texttt{PC} \neq 0}{\rets{s, \texttt{PC}, \texttt{loop}~
            e~x}{s , \texttt{PC}}} \quad
        \inference{v = \texttt{eval}~s~e & v = 0}{\goes{s, 0, \texttt{loop}~
            e~x}{s , 0, \texttt{stop}}} \\
        \inference{v = \texttt{eval}~s~e & v \neq 0}{\goes{s, 0, \texttt{loop}~
            e~x}{s, 0, x~\texttt{;;}_{c}~\texttt{loop}~e~x}} \quad
        \inference{v = \texttt{eval}~s~e & s\pr =
        \texttt{update}~s~n~v}{\rets{s,0,\texttt{assign}~n~v}{s\pr,0}}
    \end{gather*}
    \caption{Secure primitives for the \low{} language.}
    \label{fig:low+}
  \end{framed}
\end{figure}

We may now define the simple ``embedding'' natural transformation $\sigma_{E} : \Sigma
\nat \Sigma_{L}$, which maps \texttt{skip} to \texttt{stop}, assignments to
\texttt{assign}, sequencing to \texttt{;;}$_{c}$ and \texttt{while} to \texttt{loop}.

\begin{wrapfigure}{r}{0.55\textwidth}
\vspace*{-0.1cm}
\begin{tikzpicture}[node distance=2cm,thick,scale=0.9, every node/.style={scale=0.9}]
  \node (A) at (0, 0) {\texttt{while}~(\texttt{lit} 0)~p};
  \node (B) at (4, 0) {\goes{s}{s, \texttt{skip}}};
  \node (C) at (0, -1.7) {\texttt{loop}~(\texttt{lit 0})~p};
  \node (D) at (4, -1.7) {\textcolor{purple}{\rets{s,1}{s,1}}};

\draw[->] (A) edge node[below] {$\rho^{*}$} (B);
\draw[->] (A) edge node[left,black] {$\sigma^{*}_{E} \circ \Sigma^{*} b^{c}_{L}$} (C);
\draw[->][blue] (B) edge node[left,black] {$b^{c}_{L} \circ B^{c} \sigma^{*}_{E}$} (D);
\draw[->][red] (C) edge node[above,black] {$\rho^{*}_{L}$} (D);
\end{tikzpicture}
\vspace*{-0.1cm}
\caption{The coherence criterion for $(\sigma_{E},b_{L})$.}
\label{fig:cohLowS}
\end{wrapfigure}

\Cref{fig:cohLowS} shows success of the coherence criterion for the
\texttt{while} case. Since the diagram commutes for all cases, $(\sigma_{E},b_{E})$
is a map of GSOS laws between \while{} and the secure version of \low{}. This
guarantees that, remarkably, despite the presence of branching, a low-level
attacker cannot illicitly access code that is unreachable on the high-level.
Regardless, the solution is a bit contrived in that the new \low{} primitives
essentially copy what \while{} does. This is partly because the above are
complex issues involving radically different
languages but also due to the current limitations of the underlying
theory. We elaborate further on said limitations, as well as advantages and
future improvements, at~\Cref{sec:proscons}.

\section{Local state encapsulation}
\label{sec:local}

High-level programming language abstractions often involve some sort of private
state space that is protected from other objects. Basic examples include
functions with local variables and objects with private members. Low-level
languages do not offer such abstractions but when it 
comes to \emph{secure architectures}, there is some type of \emph{hardware
sandboxing}~\footnote{Examples of this are enclaves in Intel
SGX~\cite{DBLP:journals/iacr/CostanD16} and object capabilities in
CHERI~\cite{DBLP:conf/sp/WatsonWNMACDDGL15}.} to facilitate the need for
\emph{local state encapsulation}. Compilation schemes that respect
confidentiality properties have been a central subject in secure compilation
work~\cite{DBLP:journals/toplas/PatrignaniAS0CP15,DBLP:conf/csfw/AgtenSJP12,DBLP:conf/popl/FournetSCDSL13,DBLP:conf/csfw/TsampasDP19},
dating all the way back to Abadi's seminal paper~\cite{DBLP:conf/ecoopw/Abadi99}.

In this example we will explore how local state encapsulation fails due to
lack of stack clearing~\cite{DBLP:conf/csfw/TsampasDP19,
  DBLP:conf/esop/SkorstengaardDB18}. We begin by extending \while{} to support 
blocks which have their own private state, thus introducing \whileb{}. More
precisely, we add the \texttt{frame}\ST{(call)} and \texttt{return} commands
that denote the beginning and end of a new block. We also have to modify
the original behavior functor $B$ to act on a stack of stores by simply
specifying $B_{B}X = [S] -> [S] \times \mathrm{Maybe}~X$, where $[S]$
denotes a list of stores. For reasons that will become apparent later on, we
shall henceforth consider stores of a certain length, say $L$. %

\begin{figure}
  \begin{framed}
    \framegather
    \begin{gather*}
      \inference{}{\rets{m, \texttt{skip}}{m}} \quad
      \inference{v = \texttt{eval'}~m~e & m\pr = \texttt{update'}~m~l~v}
      {\rets{m, l~\texttt{:=}~e}{m\pr}} \quad
      \inference{\rets{m, p}{m\pr}}{\goes{m, p ; q}{m\pr , q}} \\
      \inference{\goes{m, p}{m\pr, p\pr}}{\goes{m, p ; q}{m\pr , p\pr ; q}} \quad
      \inference{\texttt{eval'}~m~e = 0}{\goes{m, \texttt{while}~
          e~p}{m ,
          \texttt{skip}}} \\
      \inference{\texttt{eval'}~m~e \neq 0}{\goes{m, \texttt{while}~
          e~p}{m , p ; \texttt{while}~e~p}}~
      \inference{s_{0} = [0, 0, \ldots, 0]}{\rets{m, \texttt{frame}}{s_{0} :: m}}~
      \inference{}{\rets{s :: m, \texttt{return}}{m}}
    \end{gather*}
    \caption{Semantics of the \whileb{} language.}
    \label{fig:whileb}
  \end{framed}
\end{figure}

The semantics for \whileb{} can be found in~\Cref{fig:whileb}. Command
\texttt{frame} allocates a new private store by  
appending one to the stack of stores while \texttt{return} pops the top frame
from the stack. This built-in, automatic (de)allocation of frames
guarantees that there are no traces of activity, in the form of stored values,
of past blocks. The rest of the semantics are similar to \while{}, only now
evaluating an expression and updating the state acts on a
stack of stores instead of a single, infinite store and \texttt{var} expressions
act on the active, topmost frame.

\subsection{Low-level stack}

\begin{figure}
  \begin{framed}
    \framegather
    \begin{gather*}
      \inference{m\pr =
        \texttt{update}~m~(l + L * sp)~(\texttt{evalSP}~m~sp~e)}
      {\rets{(m,sp), l~\texttt{:=}~e}{(m\pr,sp)}} \quad
      \inference{sp > 0}{\rets{(m,sp), \texttt{return}}{(m,sp-1)}} \\
      \inference{\rets{(m,sp), p}{(m\pr,sp)}}{\goes{(m,sp), p ; q}{(m\pr,sp) , q}} \quad
      \inference{\goes{(m,sp), p}{(m\pr,sp), p\pr}}{\goes{(m,sp), p ; q}{(m\pr,sp) , p\pr ; q}} \\
      \inference{}{\rets{(m,sp), \texttt{skip}}{(m,sp)}} \quad \inference{
        \texttt{evalSP}~m~sp~e = 0}{\goes{(m,sp), \texttt{while}~e~p}{(m,sp) , \texttt{skip}}} \\
      \inference{\texttt{evalSP}~m~sp~e \neq 0}{\goes{(m,sp), \texttt{while}~
          e~p}{(m,sp) , p ; \texttt{while}~e~p}} \quad
      \inference{}{\rets{(m, sp), \texttt{frame}}{(m,sp+1)}}
    \end{gather*}
    \caption{Semantics of the \whilestack{} language.}
    \label{fig:stack}
  \end{framed}
\end{figure}

In typical low-level instruction sets like the Intel x86~\cite{intel2016} or
MIPS~\cite{mips2016} there is a single, continuous \emph{memory} which is 
partitioned in \emph{frames} by using processor registers. \Cref{fig:stack} shows
the semantics of \whilestack{}, a 
variant of \whileb{} with the same syntax and which uses a simplified low-level
stack mechanism. The difference is
that the stack frames are all sized $L$, the same size as each individual store
in \whileb{}, so at each \texttt{frame} and \texttt{return} we need only increment and
decrement the \emph{stack pointer}. The presence of the stack pointer, which is
essentially a natural number, means that the behavior of \whilestack{} is $B_{S}X = S \times
\mathbb{N} -> S \times \mathbb{N} \times \mathrm{Maybe}~X$. The new evaluation function,
\texttt{evalSP}, works similarly to \texttt{eval} in~\Cref{def:eval1}, except for
\texttt{var $l$} expressions that dereference values at offset $l + L * sp$.

\subsubsection{An insecure compiler}

\whileb{} and \whilestack{} share the same syntax so we only need a behavioral
translation, which is all about relating the two different notions of stack. We
thus define natural transformation $b_B : B_{B} \nat B_{S}$:
\begin{align*}
     & b_B : \forall X.~([S_L] -> [S_L] \times \mathrm{Maybe}~X) -> S -> \mathbb{N} -> S \times \mathbb{N} \times \mathrm{Maybe}~X \\
     & b_B~f~s~sp = (\mathrm{override}~(\mathrm{join}~m)~s, \mathrm{len}~m, y)~\mathtt{where}~(m, y) = f~(\text{div}~s~sp) \\
     & \qquad \text{div}~s~sp = (\text{take}~L~s) :: (\text{div}~(\text{drop}~L~s)~(sp - 1)) \\
     & \qquad \mathrm{override}~s'~s = s'~\texttt{++}~\mathrm{drop}~(\mathrm{len}~s')~s
\end{align*}
\begin{figure}
  \begin{center}
\begin{tikzpicture}
  \node (A) at (0, 0) {\texttt{frame}};
  \node (B) at (4, 0) {\rets{{[]}}{{[s_{0}]}}};
  \node (C) at (0, -1.5) {\texttt{frame}};
  \node (D) at (4, -1.5) {\inference{\textcolor{blue}{s' = \mathrm{override}~s_0~s}}{\rets{{(s , 0)}}{({\textcolor{red}{s}/\textcolor{blue}{s'},1)}}}};

\draw[->] (A) edge node[above] {$\rho^{*}_{B}$} (B);
\draw[->] (A) edge node[left,black] {$\Sigma^{*}_{B} b^{c}_{B}$} (C);
\draw[->][blue] (B) edge node[right,black] {$b^{c}_{B}$} (D);
\draw[->][red] (C) edge node[below,black] {$\rho^{*}_{St}$} (D);
\end{tikzpicture}
\end{center}
\caption{Failure of the criterion for $(id, b_{B})$. }
\label{fig:cohStack}
\end{figure}
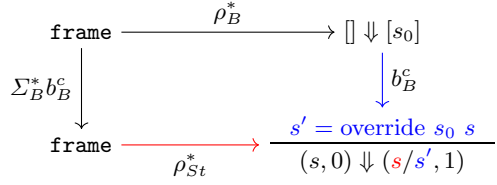

We ``divide'' an infinite list by the number of stack frames, feed the result
to the behavior function $f$ and join (``flatten'') it back together while
keeping the original part of the infinite list which extends beyond the
\emph{active stack} intact. Note that in the case of the \texttt{frame} command
$f$ adds a new frame to the list of stores. The problem is that in \whileb{} the
new frame is initialized to 0 in contrast to \whilestack{} where \texttt{frame} does not
initialize new frames. This leads to a failure of the coherence criterion for
$(id, b_{B})$ as we can see in~\Cref{fig:cohStack}. 

Failure of the criterion is meaningful in that it underlines key
problems of this compiler which can be exploited by a low-level attacker. First,
the low-level calling convention indirectly allows terms to access expired stack
frames. Second, violating the assumption in \whileb{} that new frames are
properly initialized breaks behavioral equivalence. 
For example, programs $a \triangleq \texttt{frame}~;~0~\texttt{:=}~\texttt{var}[0] + 1$ and
$b \triangleq \texttt{frame}~;~0~\texttt{:=}~1$ behave identically in \whileb{} but not in
\whilestack{}.

\subsubsection{Solution}

It is clear that the lack of stack frame initialization in \whilestack{} is the
lead cause of failure so we introduce the following fix in the \texttt{frame} rule.
\begin{align*}
  & \inference{m\pr = (\text{take}~(L * sp)~m)~\texttt{++}~s_{0}~\texttt{++}~(\text{drop}~((L+1) * sp)~m)}{\rets{(m, sp), \texttt{frame}}{(m\pr,sp+1)}}
\end{align*}

\begin{wrapfigure}{r}{0.5\textwidth}
\vspace*{-0.1cm}
\begin{tikzpicture}[node distance=2cm,thick,scale=0.9, every node/.style={scale=0.9}]
  \node (A) at (0, 0) {\texttt{frame}};
  \node (B) at (4, 0) {\rets{{[]}}{{[s_{0}]}}};
  \node (C) at (0, -1.6) {\texttt{frame}};
  \node (D) at (4, -1.6) {\inference{\textcolor{purple}{s' = s_{0}~\texttt{++}~(\text{drop}~L~s)}}{\rets{{(s , 0)}}{({\textcolor{purple}{s'},1)}}}};

\draw[->] (A) edge node[below] {$\rho^{*}_B$} (B);
\draw[->] (A) edge node[left,black] {$\Sigma^{*}_{B} b^{c}_{B}$} (C);
\draw[->][blue] (B) edge node[right,black] {$b^{c}_{B}$} (D);
\draw[->][red] (C) edge node[above,black] {$\rho^{*}_{St}$} (D);
\end{tikzpicture}
\caption{The coherence criterion for $(id, b_{B})$ under the new \texttt{frame}
  rule.}
\label{fig:cohStackSuc}
\end{wrapfigure}
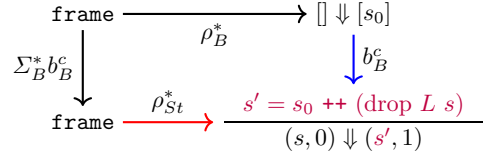

The idea behind the new \texttt{frame} rule is that the L-sized block in position
$sp$, which is going to be the new stack frame, has all its values replaced by
zeroes. As we can see in~\Cref{fig:cohStackSuc}, the coherence criterion is now
satisfied and the example described earlier no longer works.

\section{Discussion and future work}
\label{sec:proscons}

\subsubsection{On Mathematical Operational Semantics}
\label{subsec:bialg}

The cases we covered in this paper are presented using Plotkin's
Structural Operational Semantics~\cite{DBLP:journals/jlp/Plotkin04a}, yet their
foundations are deeply categorical~\cite{DBLP:conf/lics/TuriP97}. Consequently,
for one to use the methods presented in this paper, the semantics involved must
fall within the framework of distributive laws, the generality of which has
been explored in the past~\cite{DBLP:conf/ctcs/Turi97,
  DBLP:journals/entcs/Watanabe02}, albeit not exhaustively. To the best of our 
knowledge, \Cref{sec:control} and \Cref{sec:local} show the first instances of
distributive laws as low-level machines.

Bialgebraic semantics are well-behaved in that \emph{bisimilarity} is a
\emph{congruence}~\cite{DBLP:journals/iandc/GrooteV92}. We used that to show
that two bisimilar programs will remain bisimilar irrespective of the 
context they are plugged into, which is not the same as contextual equivalence.
However, full abstraction is but one of a set of proposed characterizations of secure compilation~\cite{DBLP:conf/esop/PatrignaniG19, abate2018journey} and the key intuition is
that our framework is suitable as long as bisimilarity adequately captures the
threat model. While this is the case in the examples, we can imagine situations
where the threat model is \emph{weaker} than the one implied by bisimilarity.

For example, language \whilep{} in~\Cref{sec:whilep} includes labels in its
transition structure and the underlying model is accurate in that \whilep{}
terms can manipulate said  labels. However, if we were to remove \texttt{obs} statements
from the syntax, the threat model becomes weaker than the one implied by
bisimilarity. Similarly in~\Cref{sec:control} and \low{}, where our threat model assumes
that the program counter can be manipulated by a low-level attacker. If we
impose a few constraints on the threat model, for instance by disallowing
arbitrary branching and only consider the initial program counter to be zero,
bisimilarity is suddenly too strong.

This issue can be classified as part of the broader effort towards coalgebraic
weak bisimilarity, a hard problem which has been an object of intense, ongoing
scientific
research~\cite{DBLP:journals/entcs/RotheM02,DBLP:conf/calco/Popescu09,
  DBLP:journals/lmcs/HasuoJS07, DBLP:journals/corr/Brengos13,
  DBLP:journals/ita/Rutten99, DBLP:journals/entcs/RotheM02,
  DBLP:conf/concur/BonchiPPR15}. Of particular interest is the work by
Abou-Saleh and Pattinson~\cite{DBLP:journals/entcs/Abou-SalehP11,
  DBLP:phd/ethos/AbouSaleh14} about bialgebraic semantics, where they use
techniques introduced in~\cite{DBLP:journals/lmcs/HasuoJS07} to obtain a more
appropriate semantic domain for effectful languages as a final coalgebra in the
Kleisli category of a suitable monad. This method is thus a promising avenue
towards exploring weaker equivalences in bialgebraic semantics, as long as these
can be described by a monad.

\subsubsection{On Maps of Distributive Laws}

Maps of distributive laws were first mentioned by Power and
Watanabe~\cite{DBLP:journals/entcs/PowerW99}, then elaborated as
\emph{Well-behaved translations} by
Watanabe~\cite{DBLP:journals/entcs/Watanabe02} and more recently by Klin and
Nachyla~\cite{DBLP:conf/calco/KlinN15}. 
Despite the few examples presented in~\cite{DBLP:journals/entcs/Watanabe02,
  DBLP:conf/calco/KlinN15}, this paper is the first major attempt towards
applying the theory behind maps of distributive laws in a concrete problem, let
alone in secure compilation.

From a theoretical standpoint, maps of distributive laws have remained largely
the same since their introduction. This comes despite the interesting
developments discussed in~\Cref{subsec:bialg} regarding distributive laws, which
of course are the subjects of \emph{maps} of distributive laws. We speculate the
existence of \emph{Kleisli} maps of distributive laws that guarantee
preservation of equivalences weaker than bisimilarity. We plan to develop this
notion and explore its applicability in future work.

\subsubsection{Conclusion}

It is evident that the systematic approach presented in this work may
significantly simplify proving compiler security as it involves a single,
simple coherence criterion. Explicit reasoning about program contexts is no
longer necessary, but that does not mean that contexts are irrelevant. On the
contrary, the guarantees are implicitly \emph{contextual} due 
to the well-behavedness of the semantics. Finally, while the overall usability
and eventual success of our method remains a question mark as it depends on the
expressiveness of the threat model, the body of work in coalgebraic weak
bisimilarity and distributive laws in Kleisli categories suggests that there are
many promising avenues for further progress.

\printbibliography

\end{document}